\documentclass[final,copyright,creativecommons]{eptcs}
\providecommand{\event}{HAS 2013} 
\usepackage{breakurl}             
\usepackage{amssymb}
\usepackage{graphicx}
\usepackage{amsmath}
\usepackage{stmaryrd}                       
\usepackage{dsfont}                         
\usepackage{bm}
\usepackage{mathdots}                       
\usepackage{xspace}
\usepackage[draft,inline,nomargin]{fixme}
\usepackage{graphicx}
\usepackage{color}
\usepackage{amsthm}
\newtheorem{theorem}{Theorem}
\newtheorem{definition}{Definition}
\newtheorem{proposition}{Proposition}

\newtheorem{example}{Example}
\newtheorem{remark}{Remark}

\newtheorem{lemma}{Lemma}

\theoremstyle{remark}
\newtheorem*{sketchProof}{Sketch of Proof}
\newcommand{\tn}[1]{\textnormal{#1}}        
\newcommand{\TA}{\ensuremath{\mathcal{A}}\xspace}

\newcommand{\TAsem}{\ensuremath{\llbracket\TA\rrbracket}\xspace}

\newcommand{\dynSys}{\ensuremath{\Gamma}\xspace}
\newcommand{\vectField}{\ensuremath{f}\xspace}
\newcommand{\sts}{\ensuremath{X}\xspace}
\newcommand{\stsIni}{\ensuremath{X_{0}}\xspace}
\newcommand{\dynSysAll}{\ensuremath{\Gamma=(\sts,\vectField)}\xspace}

\newcommand{\patitS}{\ensuremath{E(\mathcal{S})}\xspace}

\title{Completeness of Lyapunov Abstraction}
\author{Rafael Wisniewski\thanks{R. Wisniewski is supported by the EDGE (Efficient Distribution of Green Energy) project.}
\institute{Section of Automation \& Control\\
Aalborg University, Denmark}
\email{raf@es.aau.dk}
\and
Christoffer Sloth\thanks{C. Sloth is supported by MT-LAB, a VKR Centre of Excellence.}
\institute{Section of Automation \& Control\\
Aalborg University, Denmark}
\email{ces@es.aau.dk}
}
\def\titlerunning{Completeness of Lyapunov Abstraction}
\def\authorrunning{R. Wisniewski \& C. Sloth}
\begin{document}
\maketitle

\begin{abstract}
In this work, we continue our study on discrete abstractions of dynamical systems. To this end, we use a family of partitioning functions to generate an abstraction. The intersection of sub-level sets of the partitioning functions defines cells, which are regarded as discrete objects. The union of cells makes up the state space of the dynamical systems. Our construction gives rise to a combinatorial object - a timed automaton.  We examine sound and complete abstractions. An abstraction is said to be sound when the flow of the time automata covers the flow lines of the dynamical systems. If the dynamics of the dynamical system and the time automaton are equivalent, the abstraction is complete.

The commonly accepted paradigm for partitioning functions is that they ought to be transversal to the studied vector field. We show that there is no complete partitioning with transversal functions, even for particular dynamical systems whose critical sets are isolated critical points. Therefore, we allow the directional derivative along the vector field to be non-positive in this work. This considerably complicates the abstraction technique. For understanding dynamical systems, it is vital to study stable and unstable manifolds and their intersections. These objects appear naturally in this work. Indeed, we show that for an abstraction to be complete, the set of critical points of an abstraction function shall contain either the stable or unstable manifold of the dynamical system.
\end{abstract}

\section{Introduction}
Formal verification is used to thoroughly analyze dynamical systems. To enable formal verification of a dynamical system, the system can be abstracted by a model of reduced complexity - often a discrete model. A discrete abstraction of a dynamical system is generated by partitioning its state space into cells and determining transitions between the cells. The generation of the abstraction is quite difficult if the behaviors of the discrete and dynamical models should be equivalent (the abstraction is complete). Therefore, most works make the implicit assumption that the studied dynamical system has only one equilibrium point, and that the equilibrium point is stable. This assumption simplifies the abstraction procedure considerably, and allows the identification of partitioning functions that are transversal to solution trajectories. However, what happens if this assumption is relaxed?

Solution trajectories of a dynamical system from \cite{springerlink:10.1007/s00236-006-0037-5} are illustrated in Figure~\ref{fig:torus}. The trajectories live on the torus (the torus is obtained by identifying opposing facets with each other).
\begin{figure}[!htb]
    \centering
       \includegraphics[scale=.8]{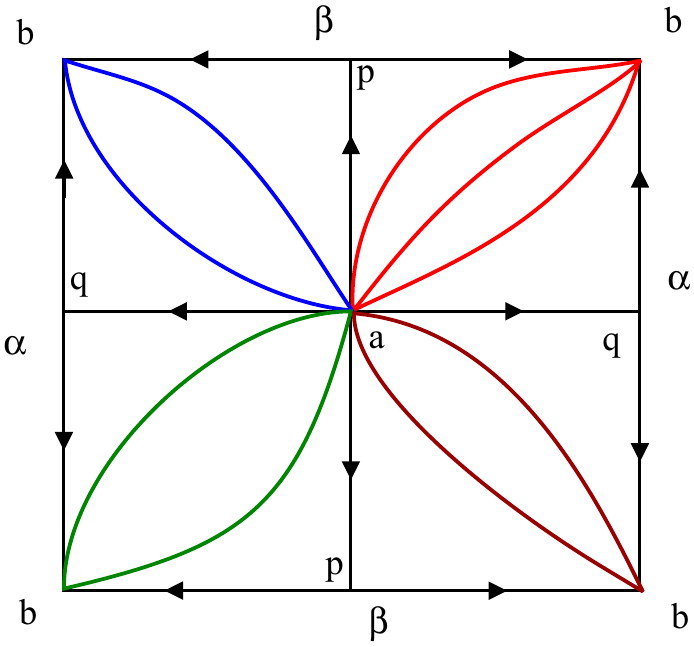}
    \caption{Flow lines of a dynamical system on a torus. The line segments between identical letters are identified to form the torus.\label{fig:torus}}
\end{figure}
The system has a stable equilibrium point ($b$), an unstable equilibrium point ($a$), and two saddle points ($\alpha$ and $\beta$). A complete abstraction cannot be generated for such system without choosing a partitioning function that has a level set that includes either the stable manifold or the unstable manifold, i.e. the vector field must be tangential to the partitioning function on this set. Furthermore, it implies that the partitioning functions must identify the stable and unstable manifolds; however, the identification of the stable and unstable manifolds is known to be a very hard problem. This implies that we cannot hope to derive algorithms for finding complete abstractions of general dynamical systems; hence, we conclude that only sound abstractions can be generated in general.

Numerous indirect verification methods for dynamical systems have been developed, i.e., verification methods based on abstracting a considered system by a model of reduced complexity, while preserving certain properties of the original system. To relate the considered system with an abstraction of it, the notions of soundness and completeness are used. Roughly speaking, the reachable set of a sound (complete) abstraction includes (equals) the reachable set of the original system. Remark that the original system and the abstraction may be of different categories. Therefore, the relation between their reachable sets should be appropriately defined, as a solution trajectory of a dynamical system cannot be directly compared to a run of an automaton. This is clarified in Section~\ref{sec:abstractions_of_dynamical_systems}.

Discrete abstractions are generated in \cite{761977} for hybrid systems and in \cite{springerlink:10.1007/s10703-007-0044-3} for continuous systems. The method developed in \cite{springerlink:10.1007/s10703-007-0044-3}, is in the class of sign based abstractions. This method relies on partitioning the state space using level sets of multiple functions and subsequently generating the transitions in the abstract model based on the Lie derivative of the partitioning function with respect to the vector field. None of the above methods generate complete abstractions; hence, the dynamics of the abstractions are only known to include the dynamics of the original system. This implies that the result of the verification may be very conservative. Other methods partition the state space by identifying positively invariant sets. In \cite{Abate20091601}, the generation of box shaped invariant sets is treated and is motivated by examples from biology.
Tools have also been developed for formal verification. An example is SpaceEx \cite{springerlink:10.1007/978-3-642-22110-1_30} that is based on finding the reachable states using support functions. This tool has good scalability properties, but does not rely on complete abstractions.

Our previous work \cite{Complete_Abstractions_of_Dynamical_Systems_by_Timed_Automata}, considers transversal functions in the partitioning of the state space. This commonly accepted transversality condition implies that there is no complete partitioning, even for a dynamical system with critical set being isolated critical points (e.g. a single saddle point).
A similar issue was detected in the abstraction of mechanical systems \cite{Abstractions_for_Mechanical_Systems}; however, the issue was resolved by only requiring the Lie derivative of the partitioning function along the vector field to be non-positive.

In this paper, we present work in progress on generating complete abstractions of more general dynamical systems. We provide simple illustrative examples that highlight the difficulties faced in the generation of abstractions. Additionally, we show that the existing conditions for generating transversal abstractions cannot trivially be extended. Throughout the paper, we consider dynamical systems without limit cycles.

The paper is organized as follows. Section~\ref{sec:preliminaries} contains preliminary definitions, Section~\ref{sec:abstractions_of_dynamical_systems} explains how a dynamical system and an abstract model are related, and Section~\ref{sec:partitioning_the_state_space} explains how the state space is partitioned using level sets of functions. Section~\ref{sec:obtaining_TA} describes how a timed automaton can be generated from the partitioning and shows that the set of critical points of a complete abstraction function contains either the stable or unstable manifold of the dynamical system, and finally Section~\ref{sec:conclusion} comprises conclusions.

\subsection{Notation}
The set $\{1,\dots,k\}$ is denoted by $\bm{k}$. $B^{A}$ is the set of maps $A\rightarrow B$. The power set of $A$ is denoted by $2^{A}$. The cardinality of the set $A$ is denoted by $|A|$. We consider the Euclidean space $(\mathds{R}^{n},\langle,\rangle)$, where $\langle,\rangle$ is the standard scalar product. $\mathds N = \{1,2, \dots\}$ is the set of natural numbers, and $\mathds Z=\{\dots,-1,0,1,\dots\}$ is the set of integers.

\section{Preliminaries}\label{sec:preliminaries}
The purpose of this section is to provide definitions related to dynamical systems and timed automata. Further details on timed automata are found in Appendix~\ref{app:ta}.

\subsection{Dynamical Systems}
A dynamical system \dynSysAll has state space $\sts\subseteq\mathds{R}^{n}$ and dynamics described by a system of ordinary differential equations $\vectField:\sts\rightarrow\mathds{R}^{n}$
\begin{align}
\dot{x}&=\vectField(x).\label{eqn:general_system}
\end{align}
We assume that \vectField has only non-degenerate singularities \cite{Introduction_to_Smooth_Manifolds} and no limit cycles.

The solution of \eqref{eqn:general_system} from an initial state $x_{0}\in \stsIni\subseteq \sts$ at time $t\geq0$ is described by the flow function $\phi_{\dynSys}:[0,\epsilon]\times X\rightarrow X$, $\epsilon>0$ satisfying
\begin{align}
\frac{d\phi_{\dynSys}(t,x_{0})}{dt}&=\vectField\left(\phi_{\dynSys}(t,x_{0})\right)\label{eqn:solution_of_auto_differential_equation}
\end{align}
for all $t\in[0,\epsilon]$ and $\phi_{\dynSys}(0,x_{0})=x_{0}$.

For a map $f: A \to B$, and a subset $C \subseteq A$, we define $f(C) \equiv \{f(x) |~x \in C \}$. Thus, the reachable set is defined as follows.
\begin{definition}[Reachable Set of Dynamical System]
The reachable set of a dynamical system \dynSys from a set of initial states $\stsIni\subseteq \sts$ on the time interval $[t_{1},t_{2}]$ is
\begin{align}
\phi_{\dynSys}([t_{1},t_{2}],X_{0}).
\end{align}
\end{definition}
A positive invariant set is defined in the following.
\begin{definition}[Positive Invariant Set]\label{def:positive_invariant_set}
Given a system $\Gamma=(X,f)$, a set $U\subseteq X$ is said to be positively invariant if for all $t\geq0$
\begin{align}
\phi_{\Gamma}(t,U)\subseteq U.
\end{align}
\end{definition}

\subsection{Timed Automata}
We use the notation of \cite{A_theory_of_timed_automata} in the definition of a timed automaton. Let $\Psi(C)$ be a set of clock constraints $\psi$ for a set of clocks $C$. This set contains all invariants and guards of the timed automaton, and is described by the following grammar
\begin{subequations}\label{eqn:clock_grammar}
\begin{align}
&\psi::=c\mathbf{\bowtie} k|\psi_{1}\bm{\wedge}\psi_{2},
\end{align}
where
\begin{align}
&c\in C,\, k\in\mathds{R}_{\geq0},\text{ and }\mathbf{\bowtie}\in\{\bm{\leq},\bm{<},\bm{=},\bm{>},\bm{\geq}\}.
\end{align}
\end{subequations}
Note that the clock constraint $k$ should be a rational number, but in this paper, no effort is made to convert the clock constraints into rational numbers. However, any real number can be approximated by a rational number with an arbitrary small error $\epsilon>0$. To make a clear distinction between syntax and semantics, the elements of $\mathbf{\bowtie}$ are bold to indicate that they are syntactic operations. 
\begin{definition}[Timed Automaton]
A timed automaton \TA is a tuple $(E,$ $E_{0},C, \Sigma, I, \Delta)$, where
\begin{itemize}
\item $E$ is a finite set of locations, and $E_{0}\subseteq E$ is the set of initial locations.
\item $C$ is a finite set of clocks.
\item $\Sigma$ is the alphabet.
\item $I:E\rightarrow\Psi(C)$ assigns invariants to locations.
\item $\Delta\subseteq E\times\Psi(C)\times\Sigma\times2^{C}\times E$ is a finite set of transition relations.  A transition relation is a tuple $(e,G_{e\rightarrow e'},\sigma,R_{e\rightarrow e'},e')$ assigning an edge between two locations, where $e$ is the source location and $e'$ is the destination location. $G_{e\rightarrow e'}\in\Psi(C)$ is the set of guards, $\sigma$ is a symbol in the alphabet $\Sigma$, and $R_{e\rightarrow e'}\subseteq C$ is a subset of clocks.
\end{itemize}
\end{definition}

A discrete flow map $\Phi_{\TA}:\mathds{R}_{\geq0}\times E_{0}\rightarrow 2^{E}$ is defined for the timed automata. Details on the flow map and its semantics are found in Appendix~\ref{app:ta}.

\section{Abstractions of Dynamical Systems}\label{sec:abstractions_of_dynamical_systems}
To evaluate the generated abstraction, it is necessary to define a relation between solution trajectories of a dynamical system and a timed automaton. We define an abstraction function, which associates subsets of the state space to locations of a timed automaton, which makes it possible to define sound and complete abstractions.

Let $\Lambda\subseteq\mathds{N}$ be an index set. An abstraction of the dynamical system \dynSysAll consists of a finite number of sets $ E \equiv \{e_{\lambda}|~\lambda \in \Lambda\}$ called cells. The cells cover the state space $X$
\[ X = \bigcup_{\lambda \in \Lambda} e_{\lambda}.\]

To the partition $E$, we associate an abstraction function, which to each point in the state space associates the cells that this point belongs to.

\begin{definition}[Abstraction Function]\label{def:abstraction_function}
Let $\Lambda\subseteq\mathds{N}$ be an index set and let $E \equiv \{e_{\lambda}|~\lambda \in \Lambda\}$ be a finite partition of the state space $\sts \subseteq \mathds{R}^{n}$. An abstraction function for $E$ is the multivalued function $\alpha_E: \sts \to 2^E$ defined by
\begin{align}
\alpha_{E}(x) \equiv \{e \in E|~x\in e\}.
\end{align}
\end{definition}

For a given dynamical system \dynSys, our aim is to simultaneously devise a partition $E$ of the state space \sts and create a timed automaton \TA with locations $E$ such that
\begin{enumerate}
\item The abstraction is \textbf{sound} on an interval $[t_1, t_2]$:
\[\alpha_E \circ \phi_{\Gamma} (t,X_0) \subseteq \Phi_{\TA} (t,\alpha_E(X_0)), \tn { for all }  t \in [t_1, t_2].\]
\item The abstraction is \textbf{complete} on an interval $[t_1, t_2]$:
\[\alpha_E \circ \phi_{\Gamma} (t,X_0) = \Phi_{\TA} (t,\alpha_E(X_0)) \tn { for all } t \in [t_1, t_2].\]
\end{enumerate}

\section{Partitioning the State Space}\label{sec:partitioning_the_state_space}
This section presents the method for partitioning the state space by functions. The partition of the state space is generated by intersecting sublevel sets of functions, and has two components: slices and cells. A slice is a sublevel set of one partitioning function, whereas a cell is a connected component of the intersection of sublevel sets of more functions.

Let $A\subseteq\mathds{R}^{n}$, then $\tn{cl}(A)$ denotes the closure of $A$. We define a slice as the set-difference of positively invariant sets.
\begin{definition}[Slice]\label{def:slice}
A nonempty set $S$ is a slice if there exist two open sets $A_{1}$ and $A_{2}$ such that
\begin{enumerate}
\item $A_{1}$ is a proper subset of $A_{2}$,
\item $A_{1}$ and $A_{2}$ are positively invariant, and
\item $S=\tn{cl}(A_{2}\backslash A_{1})$.
\end{enumerate}
\end{definition}
Since $A_{1}$ and $A_{2}$ are positively invariant sets, a trajectory initialized in $S$ can propagate to $A_{1}$, but no solution initialized in $A_{1}$ can propagate to $S$. We adopt the convention that $\emptyset$ is a positively invariant set of any dynamical systems.

%

To devise a partition of a state space, we need to define collections of slices, called slice-families.
\begin{definition}[Slice-Family]
Let $k\in\mathds{N}$ and \[A_{0}\subset A_{1}\subset\dots\subset A_{k}\] be a collection of positive invariant sets of a dynamical system \dynSysAll with $\sts \subseteq A_k$ and $A_{0}=\emptyset$. We say that the collection  \[\mathcal{S} \equiv \{S_i = \tn{cl}(A_i\backslash A_{i-1})|~i \in\bm{k}\} \] is a slice-family generated by the sets $\{A_{i}|~i\in\bm{k}\}$ or just a slice-family.
\end{definition}
We associate a function to each slice-family $\mathcal{S}$ to provide a simple way of describing the boundary of a slice. Such a function is called a partitioning function.
\begin{definition}[Partitioning Function]\label{def:part_fun}
Let \dynSysAll be a dynamical system and let $\mathcal{S}$ be a slice-family generated by the sets $\{A_{i}|~i\in\bm{k}\}$. A continuously differentiable function $\varphi:\sts\rightarrow\mathds{R}$ is a partitioning function for $\mathcal{S}$ if there is a sequence
\[a_0 < \hdots < a_k,~~a_{i}\in\mathds{R}\cup\{-\infty,\infty\}\]
where
\begin{align}
\tn{cl}(A_{i}) = \varphi^{-1}([a_{i-1},a_{i}]).
\end{align}
\end{definition}

Next, a transversal intersection of slices is defined.
\begin{definition}[Transversal Intersection of Slices]
We say that the slices $S_{1}$ and $S_{2}$ intersect each other transversally and write
\begin{align}
S_{1}\pitchfork S_{2}=S_{1}\cap S_{2}
\end{align}
if their boundaries intersect each other transversally.
\end{definition}
Cells are generated via intersecting slices. 

\begin{definition}[Extended Cell]\label{def:extended_cell}
Let $\mathcal{S}=\{\mathcal{S}^{i}|i\in\bm{k}\}$ be a collection of $k$ slice-families and let
\[\mathcal{G}(\mathcal{S})\equiv\{1,\dots,|\mathcal{S}^{1}|\}\times\dots\times\{1,\dots, |\mathcal{S}^{k}|\} \subset \mathds N^k.\] Denote the $j^{\tn{th}}$ slice in $\mathcal{S}^{i}$ by $S^{i}_{j}$ and let $g\in\mathcal{G}(\mathcal{S})$. Then
\begin{align}
e_{\tn{ex},g}&\equiv\,\pitchfork_{i=1}^{k} S^{i}_{g_{i}},\label{eqn:extended_cell_def}
\end{align}
where $g_{i}$ is the $i^{\tn{th}}$ component of the vector $g$. Any nonempty set $e_{\tn{ex},g}$ is called an extended cell of $\mathcal S$.
\end{definition}
The cells in \eqref{eqn:extended_cell_def} are called extended cells, since the transversal intersection of slices may form multiple disjoint sets in the state space; however, it is desired to have cells that are connected. 
\begin{definition}[Cell]\label{def:cell}
Let $\mathcal{S}=\{\mathcal{S}^{i}|i\in\bm{k}\}$ be a collection of $k$ slice-families. A cell $e_{(g,h)}$ of $\mathcal S$ is a connected component of an extended cell of $\mathcal S$
\begin{subequations}
\begin{align}
\bigcup_{h} e_{(g,h)} &= e_{\tn{ex},g}, \text{ where}\\
e_{(g,h)}\cap e_{(g,h')}&=\emptyset\indent\forall h\neq h'.
\end{align}
\end{subequations}
\end{definition}

A finite partition based on the transversal intersection of slices is defined in the following.
\begin{definition}[Finite Partition]\label{def:finite_partition}
Let $\mathcal{S}=\{\mathcal{S}^{i}|i\in\bm{k}\}$ be a collection of slice-families. We define a finite partition \patitS by
\begin{align}
e\in \patitS
\end{align}
if and only if $e$ is a cell of $\mathcal S$.
\end{definition}

We propose to use only functions that are nonincreasing along trajectories of the dynamical system \dynSys as partitioning functions.
\begin{definition}[Nonincreasing Partitioning Function]\label{def:lyapunov_function}
Let \sts be an open connected subset of $\mathds{R}^{n}$ and suppose that $\vectField:\sts\rightarrow\mathds{R}^{n}$ is continuous. Then a real non-degenerate differentiable function $\varphi:\sts\rightarrow\mathds{R}$ is said to be a partitioning function for \vectField if
\begin{subequations}
\begin{align}
&\psi(x)\equiv\sum_{j=1}^{n}\frac{\partial \varphi}{\partial x_{j}}(x)\vectField_{j}(x)\label{eqn:Lyap_der}\\
&\psi(x)\leq0~~~\forall x\in \sts.\label{eqn:Lyap_pos_definite}
\end{align}
\end{subequations}
\end{definition}

Figure~\ref{fig:extended_cell} shows a state space partitioned by two partitioning functions, with level sets illustrated by red respectively blue lines. The intersection of slices generates an extended cell (gray area), consisting of four cells (each connected component of the gray area).
\begin{figure}[!htb]
    \centering
       \includegraphics[scale=1]{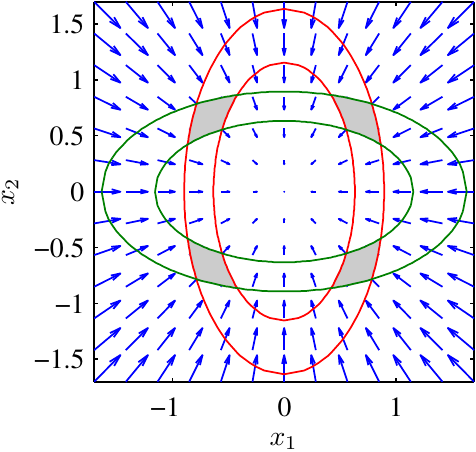}
    \caption{Illustration of a vector field in a state space, partitioned using level sets of two partitioning functions.\label{fig:extended_cell}}
\end{figure}



In the next section, we show how the partitioning functions should be chosen to generate a complete abstraction.

\section{Generation of Timed Automaton from Finite Partition}\label{sec:obtaining_TA}
A timed automaton \TA is generated from a finite partition \patitS as follows.

\begin{definition}[Generation of Timed Automaton]\label{def:generation_of_TA}
Given  a finite collection of slice-families $\mathcal{S}=\{\mathcal{S}^{i}|~i\in\bm{k}\}$, and pairs of times $\mathcal T = \{(\underline{t}^{i}_{g_{i}}, \overline{t}^{i}_{g_{i}}) \in \mathds R_{\geq0}^2|~i\in\bm{k},\,g_{i}\in\{1,\dots,|\mathcal{S}^{i}|\} \}$.
We define the timed automaton $\mathcal{A}(\mathcal S, \mathcal T)=(E, E_{0}, C, \Sigma, I, \Delta)$ by
\begin{itemize}
\item \textbf{Locations:} The locations of $\mathcal{A}$ are given by
\begin{align}
E=E(\mathcal{S}).\label{eqn:TA_locations}
\end{align}
This means that a location $e_{(g,h)}$ is identified with the cell $e_{(g,h)}=\alpha_{E(\mathcal S)}^{-1}(\{e_{(g,h)}\})$ of the partition $E(\mathcal{S})$, see Definition~\ref{def:abstraction_function}.

\item \textbf{Clocks:} The set of clocks is
$C = \{c^{i}|~i\in\bm{k}\}$.

\item \textbf{Alphabet:} The alphabet is
$\Sigma = \mathcal \{\sigma^{i}|~i\in\bm{k}\}$.

\item \textbf{Invariants:} In each location $e_{(g,h)}$, we impose an invariant
\begin{align}
I(e_{(g,h)})&= \bigwedge_{i=1}^{k}c^{i}\bm{\leq} \overline{t}^{i}_{g_{i}}.\label{eqn:TA_invariants}
\end{align}

\item \textbf{Transition relations:} If a pair of locations $e_{(g,h)}$ and $e_{(g',h')}$ satisfy the following two conditions
\begin{enumerate}
\item $e_{(g,h)}$ and $e_{(g',h')}$ are adjacent; that is $e_{(g,h)}\cap e_{(g',h')}\neq\emptyset$, and
\item $g'_{i}\leq g_{i}$ for all $i\in\bm{k}$.
\end{enumerate}
Then there is a transition relation
\begin{subequations}
\begin{align}
\notag\delta_{(g,h)\rightarrow (g',h')}& = (e_{(g,h)}, G_{(g,h)\rightarrow(g',h')},\sigma^i, R_{(g,h)\rightarrow(g',h')}, e_{(g',h')}),
\end{align}
where
\begin{align}
&G_{(g,h)\rightarrow(g',h')}=\bigwedge_{i=1}^{k}\begin{cases}c^{i}\geq \underline{t}^{i}_{g_{i}}&\tn{if }g_{i}-g'_{i}=1\\c^{i}\geq0&\tn{otherwise.}\end{cases}\label{eqn:TA_guard}
\end{align}
Note that $g_{i}-g'_{i}=1$ whenever a transition labeled $\sigma^{i}$ is taken.

Let $i\in\bm{k}$. We define $R_{(g,h)\rightarrow(g',h')}$ by
\begin{align}
&c^{i}\in R_{(g,h)\rightarrow(g',h')}\label{eqn:TA_reset}
\end{align}
iff $g_{i}-g'_{i}=1$.
\end{subequations}
\end{itemize}
\end{definition}

To ensure that the properties of an abstraction is not only valid for a particular choice of level sets, we impose the derived condition for any choice of level set in the partition.
%
%

From Definition~\ref{def:generation_of_TA}, it is seen that to generate a timed automaton, it is required to devise a partition of the state space, and find a set of invariant and guard conditions. Therefore, we provide a condition under which an abstraction is complete, recall the definition of a complete abstraction in Section~\ref{sec:abstractions_of_dynamical_systems}.
\begin{proposition}[\cite{Complete_Abstractions_of_Dynamical_Systems_by_Timed_Automata}]\label{prop:completeness_of_partition}
Given a dynamical system \dynSysAll, a collection of partitioning functions $\{\varphi^{i}|~i\in\bm{k}\}$, a collection of values $\{a^{i} _{g_{i}}\in\mathds{R} |~i\in\bm{k},\,g_{i}\in\{1,\dots,|\mathcal{S}^{i}|\} \}$ generating $\mathcal{S}$, and pairs of times $\mathcal T = \{(\underline{t}^{i}_{g_{i}}, \overline{t}^{i}_{g_{i}}) |~i\in\bm{k},\,g_{i}\in\{1,\dots,|\mathcal{S}^{i}|\} \}$. The timed automaton $\TA(\mathcal{S}, \mathcal{T})$ is a complete abstraction of $\Gamma$ if and only if for any $i \in \bm{k}$
\begin{enumerate}
\item for any pair of regular values $(a^{i}_{g_{i}-1},a^{i}_{g_{i}})$ with $g\in\mathcal{G}(\mathcal{S})$ (see the definition of $\mathcal{G}(\mathcal{S})$ in Definition~\ref{def:extended_cell}) there exists a time $t^{i}_{g_{i}}$ such that for all $x_{0}\in (\varphi^{i})^{-1}(a^{i}_{g_{i}})$
\begin{align}
\phi_{\Gamma}(t^{i}_{g_{i}},x_{0})\in (\varphi^{i})^{-1}(a^{i}_{g_{i}-1})\label{eqn:suf_cond_sound2}
\end{align}
and
\item $\overline{t}_{S^{i}_{g_{i}}}=\underline{t}_{S^{i}_{g_{i}}}=t^{i}_{g_{i}}$.
\end{enumerate}
\end{proposition}
We say that $\{\varphi^i|~i\in\bm{k}\}$ generates a complete abstraction if there exist times $\mathcal{T}$ such that Proposition~\ref{prop:completeness_of_partition} is satisfied.

In the introduction of the paper, we claim that it is impossible to generate a complete abstraction of a system with a saddle point using transversal partitioning functions. In the following example a complete abstraction is given for such a system by nonincreasing partitioning functions.

\begin{example}\label{ex:saddle_lin}
Consider the following two-dimensional linear vector field
\begin{align}
\begin{bmatrix}
\dot{x}_1\\
\dot{x}_2
\end{bmatrix}&=
\begin{bmatrix}
-x_1\\
x_2
\end{bmatrix}\label{eqn:vectFieldEx1}.
\end{align}
The system has a saddle point and its phase plot is shown in Figure~\ref{fig:example1_saddle_phase_plot}.
\begin{figure}[!htb]
    \centering
       \includegraphics[scale=1]{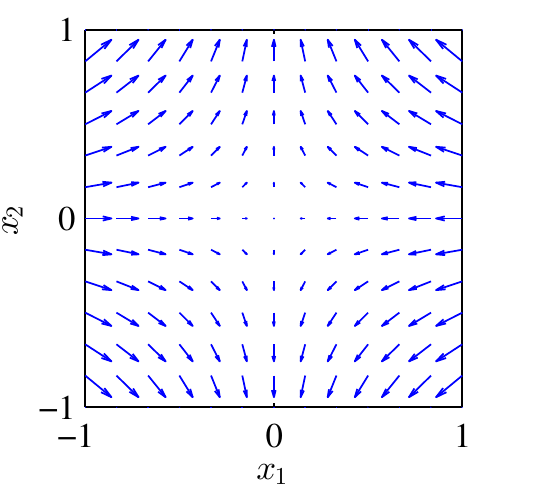}
    \caption{Phase plot of a linear system with a saddle point.\label{fig:example1_saddle_phase_plot}}
\end{figure}

We choose the following partitioning functions
\begin{subequations}
\begin{align}
\varphi_{1}(x_1,x_2) &= x_1^2,\\
\varphi_{2}(x_1,x_2) &= -x_2^2,
\end{align}
\end{subequations}
Recall from \eqref{eqn:Lyap_der} that we denote Lie derivatives by $\psi$. The Lie derivatives of the partitioning functions along the vector field in \eqref{eqn:vectFieldEx1} become
\begin{subequations}
\begin{align}
\psi_{1}(x_1,x_2) &= -2x_1^2,\\
\psi_{2}(x_1,x_2) &= -2x_2^2.
\end{align}
\end{subequations}
This implies that the abstraction generated by $\{\varphi_{1},\varphi_{2}\}$ is complete. Completeness is concluded from Proposition~\ref{prop:completeness_of_partition}.
\end{example}

Proposition~\ref{prop:completeness_of_partition} does not provide a straightforward method for computing a complete abstraction, as the conditions are not numerically tractable. Therefore, we rephrase \eqref{eqn:suf_cond_sound2} as a relation between the level sets of the partitioning function and its derivative along the vector field.
It is difficult to determine if the partitioning functions in Example~\ref{ex:saddle_lin} generate a complete abstraction from Proposition~\ref{prop:completeness_of_partition}; however, this is clear from the following proposition, as the level sets of $\phi$ and $\psi$ coincide.
\begin{proposition}\label{prop:prop2}
Let \dynSys be a dynamical system and $\varphi$ be a smooth function, in particular a partitioning function. Then for any regular value $a\in\mathds{R}$, the following statements are equivalent
\begin{enumerate}
\item there exists $b\in\mathds{R}$ such that \begin{align}
\{x\in\mathds{R}^{n}|~\varphi(x)-a=0\}\subseteq\{x\in\mathds{R}^{n}|~\psi(x)-b=0\},\label{eqn:subset_lev_set11}
\end{align}
where
\begin{align}
\notag\psi(x) \equiv d \varphi (f) (x) = \sum_{i}\frac{\partial\varphi}{\partial x_{i}}(x)f_{i}(x).
\end{align}
\item there exists $\epsilon>0$ such that for any $t\in(-\epsilon,\epsilon)$
\begin{align}
\varphi(\phi_{\Gamma}(t,x_{1}))=\varphi(\phi_{\Gamma}(t,x_{2}))~~~~\forall x_{1},x_{2}\in \varphi^{-1}(a).\label{eqn:new2}
\end{align}
\end{enumerate}
\end{proposition}
\begin{proof}
We show that 2) implies 1).
We differentiate both sides of
\begin{align*}
\varphi(\phi_{\Gamma}(t,x_{1}))=\varphi(\phi_{\Gamma}(t,x_{2}))
\end{align*}
with respect to $t$
\begin{subequations}
\begin{align}
\sum_{i}\frac{\partial\varphi}{\partial x_{i}}(\phi_{\Gamma}(t,x_{1}))f_{i}(\phi_{\Gamma}(t,x_{1}))&=
\sum_{i}\frac{\partial\varphi}{\partial x_{i}}(\phi_{\Gamma}(t,x_{2}))f_{i}(\phi_{\Gamma}(t,x_{2})),\\
\psi(\phi_{\Gamma}(t,x_{1}))&=\psi(\phi_{\Gamma}(t,x_{2})).
\end{align}
\end{subequations}
At $t=0$, $\psi(x_{1})=\psi(x_{2})$. Hence, \eqref{eqn:subset_lev_set11} is satisfied.

To show that 1) implies 2), we define a convenient state transformation inspired by \cite[p.~13]{Morse_Theory_book}. In the new coordinates, the vector field has only one nonzero component.

Let $M=\varphi^{-1}(a)$ be a smooth manifold. By Sard's Theorem, there exists an open neighborhood $U$ of $a$ such that any point in $U$ is a regular value of $\varphi$ \cite[p.~132]{Introduction_to_Smooth_Manifolds}. Define the smooth function $\eta:\varphi^{-1}(U) \rightarrow \mathds{R}$ by
\begin{align}
\eta \equiv \frac{1}{||\nabla\varphi||^2},
\end{align}
where $\nabla\varphi$ is the gradient of $\varphi$ (with respect to a Riemannian metrics $\left<\cdot, \cdot \right>$ on M). The function $\eta$ is well defined, since $\varphi^{-1}(U) \subset M$ is an (open) set of regular points of $\varphi$. Define the vector field $\xi$ on $\varphi^{-1}(U)$ by
\begin{align}
\xi \equiv \eta \nabla \varphi.
\end{align}
The derivative of $\varphi$ in the direction of $\xi$ is
\begin{align}
d \varphi(\xi)=\langle\nabla\varphi,\xi\rangle=\eta\langle\nabla\varphi,\nabla\varphi\rangle=1,
\label{unit_increment}
\end{align}
Choose $a', a'' \in \mathds R$ such that $a'<a<a''$ and $[a',a'']\subset U$. We define the map $F(\cdot,\cdot):M\times[a',a'']\rightarrow\varphi^{-1}([a',a'']) \subset X$, where $F(x_{0},t)$ is the solution of $\xi$ from initial state $x_{0}$. From \eqref{unit_increment}, we have
\begin{align}
t\mapsto\varphi\circ F(x_{0},t)=a+t.\label{eqn:good_map}
\end{align}

We represent the vector field $f$ and function $\varphi$ in new coordinates $q = F(x,t)$
\begin{subequations}
\begin{align}
\tilde f (q) & = (D F(q))^{-1}f \circ F(q) \\
\tilde \varphi (q) & = \varphi\circ F(q) \label{eqn:varphi_tilde}
\end{align}
\end{subequations}

Notice that the vector fields $f$ and $\tilde f$ are $F$-related. The differential of $\tilde \varphi$ is
\begin{align}
d\tilde{\varphi}(q)=d\varphi(F(q))DF(q).
\end{align}
Hence, the derivative of $\tilde \varphi$ if the direction of $\tilde f$ is
\begin{align}
d\tilde{\varphi}(\tilde{f})&=d \varphi \circ F ~ DF (DF)^{-1} f \circ F\\
&=d\varphi(f) \circ F.
\end{align}

Let $x_{1},x_{2}\in M$, then $F(x_1,t)$ and $F(x_2,t)$ are regular points for $t \in (-\epsilon, \epsilon)$. By \eqref{eqn:subset_lev_set11}
\begin{align}
d \varphi (f)(F(x_{1},t)) = d\varphi (f)(F(x_{2},t)).
\end{align}
We define $\tilde{\psi}=d\tilde{\varphi}(\tilde{f})$, and conclude that
\begin{align}
\tilde{\psi}(x_1,t) = \tilde{\psi}(x_2,t)
\end{align}

From~\eqref{eqn:good_map}, $\tilde{\varphi}$ only depends on the last coordinate. Therefore, denoting $f = (f_1, \hdots, f_n)$, we have
\begin{align}
\tilde{\psi} = d\tilde{\varphi}\tilde{f}=\tilde{f}_{n}.
\end{align}

Since $\tilde \psi (x_1,t) = \tilde \psi(x_2,t)$ for any pair $x_1, x_2 \in M$ and $t \in [a',a'']$, we have $\tilde f_n(x_1,t) = \tilde f_n(x_2,t)$. In other words, the $n$th component of the vector field $\tilde f$ depends only on its last $n$ coordinate $t$. As a consequence, denoting the flow map of the vector field $\tilde f$ by $\phi_{\tilde \Gamma}$, we have
\begin{align}
\phi_{\tilde \Gamma} ((x_1,0),t ) = \phi_{\tilde \Gamma} ((x_2,0),t ) \in M \times (-\epsilon, \epsilon).
\end{align}
The vector field $f$ and $\tilde f$ are $F$-related, hence also
\begin{align}
\phi_{\Gamma} (x_1,t ) = \phi_{\Gamma} (x_2,t ).
\end{align}
Thus,  the inclusion \eqref{eqn:suf_cond_sound2} holds.
\end{proof}

Unfortunately, we cannot relax proposition to include critical values as well, i.e., it does not hold for any $a\in\mathds{R}$. This is clarified in the following, by presenting a dynamical system with state space $\sts\subseteq\mathds{R}^2$ for which there exists no complete abstraction generated by transversal partitioning functions.

According to Proposition~\ref{prop:completeness_of_partition}, it takes the same time for two trajectories to propagate between level sets of complete partitioning functions. Consider a dynamical system with flow lines shown in Figure~\ref{fig:vectorFieldmm}.
\begin{figure}[!htb]
    \centering
       \includegraphics[scale=1.2]{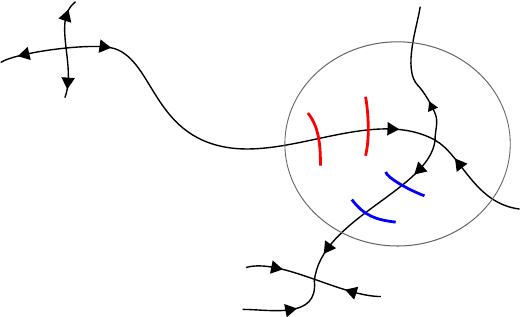}
    \caption{Flow lines of a dynamical system, with a saddle point, a stable equilibrium point, and an unstable equilibrium point.\label{fig:vectorFieldmm}}
\end{figure}
Level sets of two partitioning functions are illustrated at the saddle point inside the gray circle. From the red level set, one trajectory goes to the saddle point; however, the remaining trajectories diverge from the saddle point. Therefore, it cannot take the same time to propagate between any two level sets, if the partitioning function is decreasing along the vector field.
In the following, we denote the stable (unstable) manifold by $W^{\tn{s}}(p)$ ($W^{\tn{u}}(p)$) \cite{hirsch}.
\begin{lemma}\label{lem:crit}
Let $\varphi$ be a partitioning function generating a complete abstraction of \dynSysAll and let $p$ be a singular point of \vectField. Then $p$ is a critical point of $\varphi$.
\end{lemma}
\begin{proof}
Suppose that $p$ is a regular point of $\varphi$ with regular value $c$. Then $\varphi^{-1}(c)$ is an $n-1$ dimensional manifold. Without loss of generality, we assume that $p$ is not stable. Let $x\in\varphi^{-1}(c)\backslash W^{\tn{s}}(p)$. Then there exists $\tau>0$ such that  $\phi(\tau,x)\notin\varphi^{-1}(c)$, but $p=\phi(\tau,p)\in\varphi^{-1}(c)$, since $p$ is a singular point of $f$. This contradicts completeness.
\end{proof}

The following theorem is the main contribution of the paper, showing that complete abstractions identify stable and unstable manifolds.
\begin{theorem}\label{thm:res}
Let $\tn{Reg}(\varphi)$ denote the set of regular values of $\varphi$, generating a complete abstraction of \dynSysAll, let \sts be a connected compact manifold, and let $p$ be a singular point of \vectField. Suppose that $\varphi^{-1}(\tn{Reg}(\varphi))\cap W^{\tn{s}}(p)\neq\emptyset$. Then 
\[W^{u}(p)\subseteq\varphi^{-1}(\varphi(p))  .\]
\end{theorem}
\begin{sketchProof}
Let $Z\equiv\varphi^{-1}(\tn{Reg}(\varphi))\cap W^{\tn{s}}(p)$ and let $R(x)\equiv\varphi^{-1}(\varphi(x))$. Suppose that there exists a point $x\in Z$ such that $Y(x)\equiv R(x)\backslash W^{\tn{s}}(p)\neq\emptyset$.

Per definition of $W^{\tn{s}}(p)$, $\lim_{t\rightarrow\infty}\phi(x,t)= p$. Let $y\in Y(x)$ then there exists a singular point $p'\neq p$ such that $y\in W^{\tn{s}}(p')$. Note that any solution goes to a singular point, as the state space \sts is compact.

Since $\varphi(x)=\varphi(y)$, the abstraction is complete, and the flow map $\phi$ is continuous, we conclude from Proposition~\ref{prop:prop2} that $\varphi(p)=\varphi(p')$. Thus, for any $z\in W(p,p')\equiv W^{\tn{u}}(p)\cap W^{\tn{s}}(p')$ the value of $\varphi$ is constant, i.e., $\varphi(p)=\varphi(z)$. Therefore, $W(p,p')\subseteq R(p)$.

To show that $W^{\tn{u}}(p)\subseteq R(p)$, we exploit that any $z\in R(x)$, $z\in W^{\tn{s}}(p'')$ for some singular point $p''$. Therefore, $\varphi(p)=\varphi(p'')$ for any such singular point.
We use notation $p\preceq p'$ iff there is a flow line from $p$ to $p'$. Whenever there is a sequence $\{p_1,\dots,p_n\}$ of singular points of the vector field $f$ such that
\[p=p_1\preceq p_2\preceq\dots\preceq p_n,\]
then $\varphi(p_i)=\varphi(p)$ for all $i\in\{1,\dots,n\}$.
Such a scenario is illustrated in Figure~\ref{fig:multi_sing_proof}.
\begingroup%
  \makeatletter%
  \providecommand\color[2][]{%
    \errmessage{(Inkscape) Color is used for the text in Inkscape, but the package 'color.sty' is not loaded}%
    \renewcommand\color[2][]{}%
  }%
  \providecommand\transparent[1]{%
    \errmessage{(Inkscape) Transparency is used (non-zero) for the text in Inkscape, but the package 'transparent.sty' is not loaded}%
    \renewcommand\transparent[1]{}%
  }%
  \providecommand\rotatebox[2]{#2}%
  \ifx\svgwidth\undefined%
    \setlength{\unitlength}{396bp}%
    \ifx\svgscale\undefined%
      \relax%
    \else%
      \setlength{\unitlength}{\unitlength * \real{\svgscale}}%
    \fi%
  \else%
    \setlength{\unitlength}{\svgwidth}%
  \fi%
  \global\let\svgwidth\undefined%
  \global\let\svgscale\undefined%
  \makeatother%
  \begin{figure}$ $%
    \put(0,0){\includegraphics[width=\unitlength]{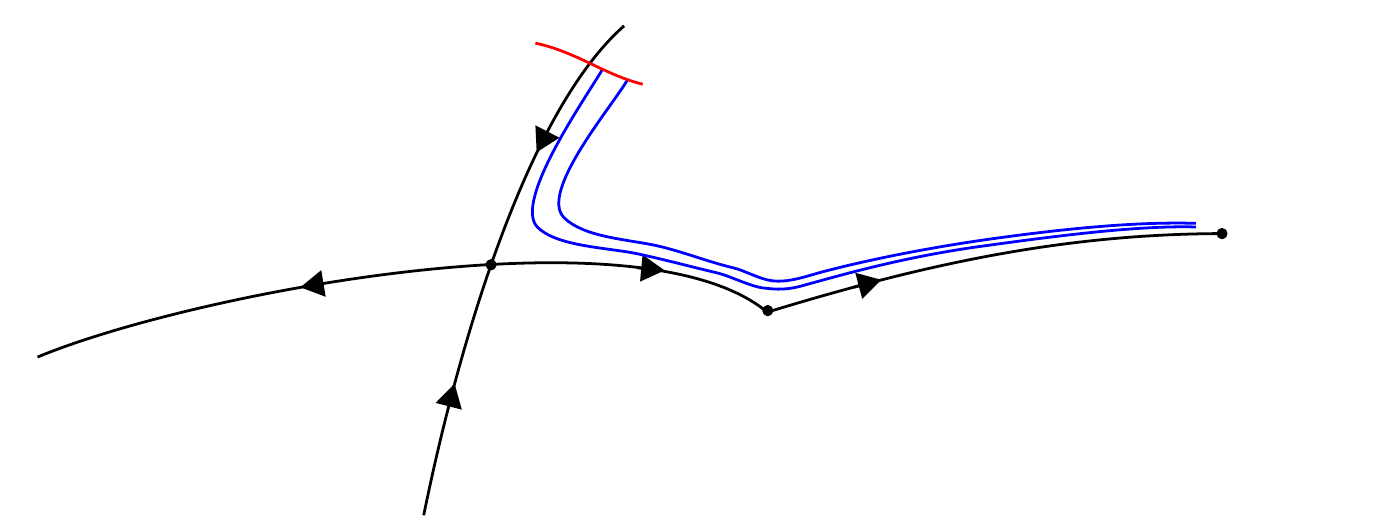}}%
    \put(0.45894143,0.35794543){\color[rgb]{0,0,0}\makebox(0,0)[lb]{\smash{$W^{\tn{s}}(p_1)$}}}%
    \put(0.35988045,0.16176776){\color[rgb]{0,0,0}\makebox(0,0)[lb]{\smash{$p_1$}}}%
    \put(0.56116808,0.12436601){\color[rgb]{0,0,0}\makebox(0,0)[lb]{\smash{$p_2$}}}%
    \put(0.88746961,0.18442785){\color[rgb]{0,0,0}\makebox(0,0)[lb]{\smash{$p_3$}}}%
    \put(0.47248211,0.31349297){\color[rgb]{0,0,0}\makebox(0,0)[lb]{\smash{$\varphi^{-1}(\varphi(x))$}}}%

\caption{State space of a system with singular points $p$, $q_1$, and $q_2$. No solution initialized on the red level set reaches $q_1$, but solutions get arbitrarily close to $q_1$.}
\label{fig:multi_sing_proof}
  \end{figure}
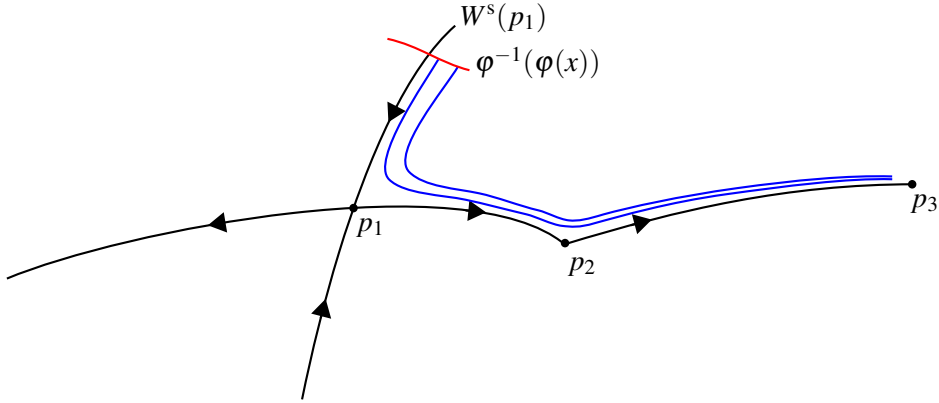%
\endgroup
Thus, we have
\[W^{\tn{u}}(p)=\bigcup_{p\preceq\dots\preceq p'}W(p,p')\subseteq R(p).\]

If $Y(x)=\emptyset$ for all $x\in Z$, then
\[\bigcup_{x\in Z} R(x)\subseteq W^{\tn{s}}(p).\]
The dimension of $R(x)$ is $n-1$. Furthermore, since $R(x)$ is a closed compact manifold and the vector field is transversal to $R(x)$, there exists $\tau>0$ such that
\[\phi(R(x),(-\tau,\tau))\]
that is an embedded manifold of dimension $n$. Thereby $\tn{dim}(W^{\tn{s}}(p))=n$, and $W^{\tn{u}}(p)=\{p\}\subseteq R(p)$.
\qed
\end{sketchProof}
One could think that the inclusion $W^{u}(p)\subseteq \varphi^{-1}(\varphi(p))$ in Theorem~\ref{thm:res} could be replaced by an equality; however, the following counterexample demonstrates that this cannot be done.
\begin{example}
Consider the following two dimensional linear vector field from Example~\ref{ex:saddle_lin}
\begin{align}
\begin{bmatrix}
\dot{x}_1\\
\dot{x}_2
\end{bmatrix}&=
\begin{bmatrix}
-x_1\\
x_2
\end{bmatrix}.
\end{align}
The system has a saddle point $p=(0,0)$ and its phase plot is shown in Figure~\ref{fig:example1_saddle_phase_plot}. We modify the partitioning function $\varphi_{1}(x_1,x_2)=x_1^2$ from Example~\ref{ex:saddle_lin}, to obtain a case where
\[W^{u}(p)\subsetneq\tilde{\varphi}_1^{-1}(\tilde{\varphi}_1(p)).\]
For this purpose, we construct a bump function, according to \cite{An_Introduction_to_Manifolds}. Let
\begin{align}
\notag f(t)&=\begin{cases}e^{-1/t}&\text{for }t>0\\0&\text{for }t\leq0\end{cases}\\
\notag g(t)&=\frac{f(t)}{f(t)+f(1-t)}.
\end{align}
From $f$ and $g$, we choose the partitioning function
\begin{align}
\tilde{\varphi}_{1}(x_1,x_2) &= g\left(\frac{x_1^2-a^2}{b^2-a^2}\right)\cdot x_1^2\label{eqn:bump_varphi1},
\end{align}
with $a=0.5$ and $b=2$. The graphs of $g\left(\frac{x_1^2-a^2}{b^2-a^2}\right)$, $\tilde{\varphi}_{1}(x_1,x_2)$, and $\varphi_{1}(x_1,x_2)$ are shown in Figure~\ref{fig:bump_example}.
\begin{figure}[!htb]
    \centering
       \includegraphics[scale=1]{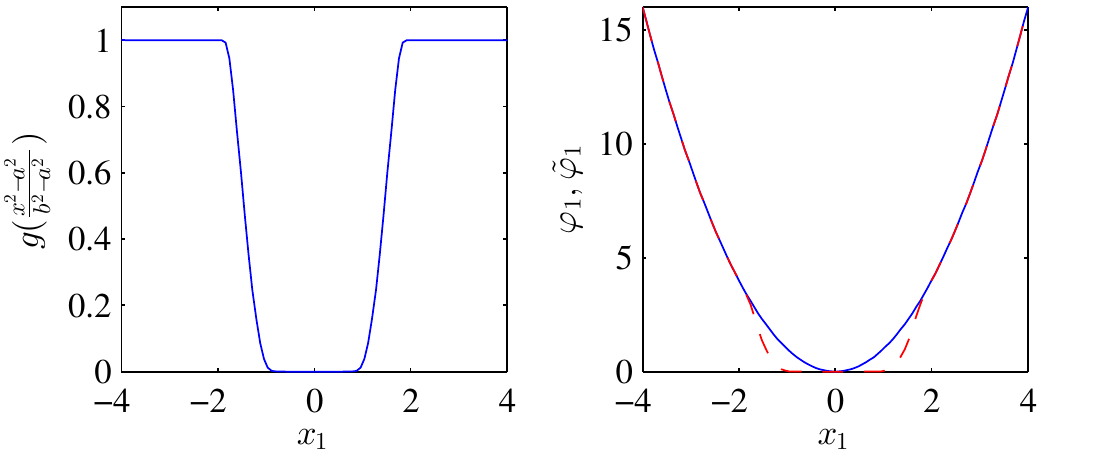}
    \caption{The left subplot shows the graph of $g\left(\frac{x_1^2-a^2}{b^2-a^2}\right)$ and the right subplot shows the graphs of $\varphi_{1}(x_1,x_2)$  (blue) and $\tilde{\varphi}_{1}(x_1,x_2)$ (dashed red).\label{fig:bump_example}}
\end{figure}

The partitioning function shown in \eqref{eqn:bump_varphi1} generates a complete abstraction and for this particular function $W^{u}(p)$ is a proper subset of $\tilde{\varphi}_1^{-1}(\tilde{\varphi}_1(p))$.
\end{example}

Theorem~\ref{thm:res} is vital in the field of abstracting generic dynamical systems, as it shows that a partitioning function generating a complete abstraction has to be constant on the unstable manifolds; hence, finding such a function is as difficult as finding the stable and unstable manifolds. This is practically impossible for systems of dimension greater than three. As an example, on all black lines in Figure~\ref{fig:morse_complex} a complete partitioning function must be constant.
\begin{figure}[!htb]
    \centering
       \includegraphics[scale=0.8]{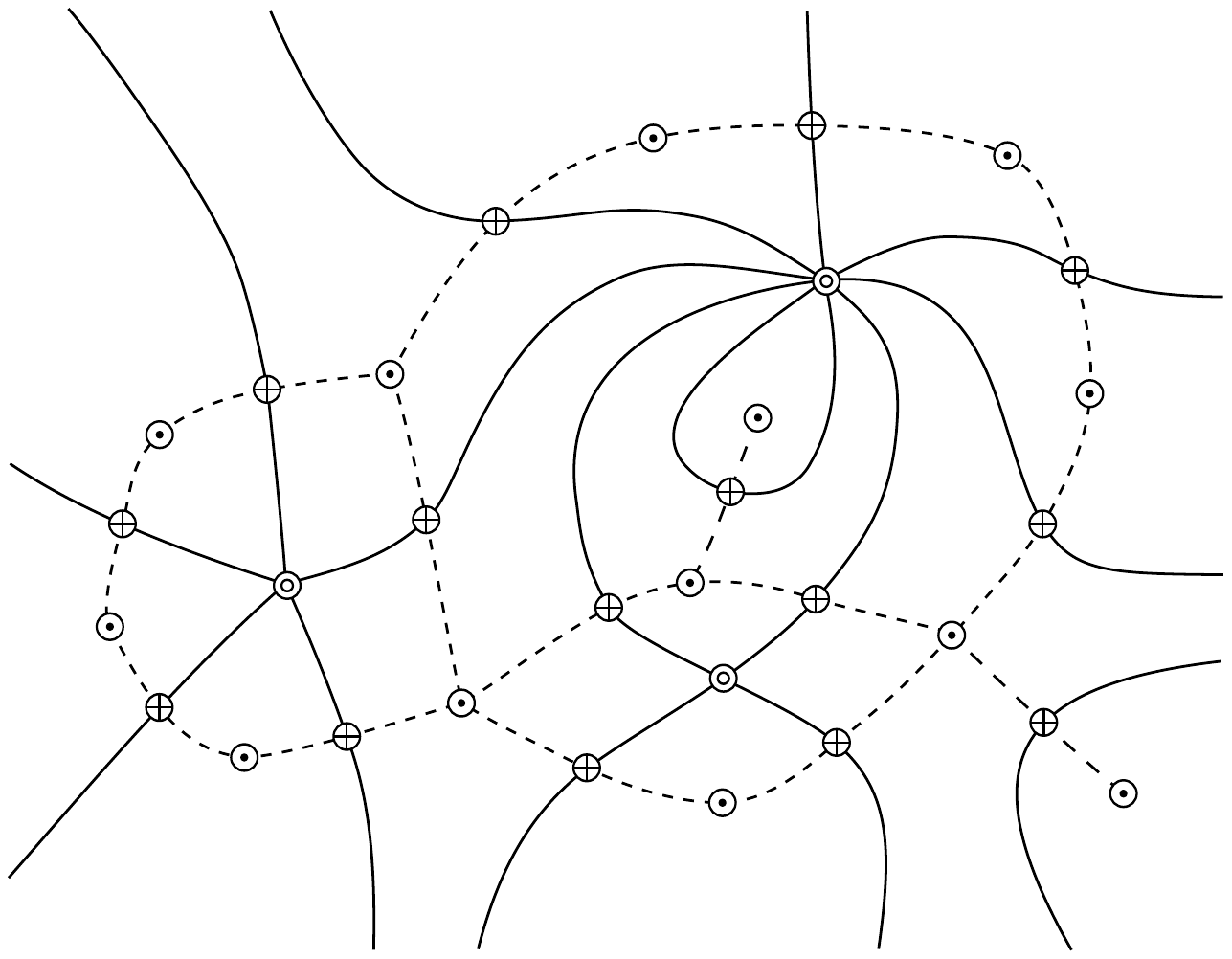}
    \caption{Morse complex, where a "+" indicates a saddle, a "$\cdot$" indicates a maximum, and a "$\circ$" indicates a minimum. The lines connecting the singular points are stable and unstable manifolds.\label{fig:morse_complex}}
\end{figure}
In conclusion, research in the field should be focused on method development for generating sound rather than complete abstractions.

\section{Conclusion}\label{sec:conclusion}
This paper highlights some issues with the use of transversal partitions. Complete abstractions of general dynamical systems cannot be generated using transversal partitioning functions. This issue arises as a partitioning function generating a complete abstraction has to be constant on the unstable manifolds.

The identification of stable and unstable manifolds is a very difficult problem. Therefore, it will not be possible to derive algorithms for generating complete abstraction of generic dynamical systems.

We have derived a theorem stating that the set of critical points of partitioning functions generating a complete abstraction must contain either the stable or unstable manifold of the dynamical system. This makes algorithmic generation of complete abstractions practically impossible for general dynamical systems.

\bibliographystyle{eptcs}
\bibliography{bibliography}

\appendix
\section{Timed Automata Details}\label{app:ta}
%

The semantics of a timed automaton is defined in the following.
\begin{definition}[Clock Valuation]\label{def:clock_valuation}
A clock valuation on a set of clocks $C$ is a mapping $v:C\rightarrow\mathds{R}_{\geq0}$. The initial valuation $v_{0}$ is given by $v_{0}(c)=0$ for all $c\in C$. For a valuation $v$, a scalar $d\in\mathds{R}_{\geq0}$, and $R\subseteq C$, the valuations $v+d$ and $v[R]$ are defined as
\begin{subequations}
\begin{align}
(v+d)(c) &= v(c)+d,\label{eqn:clock_valuation_delay}\\
v[R](c) &= \begin{cases}0&\text{for }c\in R,\\v(c)&otherwise.\end{cases}\label{eqn:clock_valuation_reset}
\end{align}
\end{subequations}
\end{definition}
It is seen that \eqref{eqn:clock_valuation_delay} is used to progress time and \eqref{eqn:clock_valuation_reset} is used to reset the clocks in the set $R$ to zero.

We denote the set of maps $v:C\rightarrow\mathds{R}_{\geq0}$ by $\mathds{R}^{C}_{\geq0}$.
\begin{remark}\label{rem:topology_valuation}
This notation indicates that we identify a valuation $v$ with $C$-tuples of nonnegative reals in $\mathds{R}_{\geq0}^{|C|}$, where $|C|$ is the number of elements in $C$. We impose the Euclidian topology on $\mathds{R}^{C}_{\geq0}$.
\end{remark}

\begin{definition}[Semantics of Clock Constraint]
A clock constraint in $\Psi(C)$ is a set of clock valuations $\{v:C\rightarrow\mathds{R}_{\geq0}\}$ given by
\begin{subequations}
\begin{align}
\llbracket c\mathbf{\bowtie} k\rrbracket &= \{v:C\rightarrow\mathds{R}_{\geq0}|v(c)\bowtie k\}\\
\llbracket\psi_{1}\bm{\wedge}\psi_{2}\rrbracket &= \llbracket\psi_{1}\rrbracket\cap\llbracket\psi_{2}\rrbracket.
\end{align}
\end{subequations}
\end{definition}
For convenience, we denote $v\in\llbracket\psi\rrbracket$ by $v\models\psi$ and denote the transition $(e,v,\sigma,e',v')$ by $(e,v)\overset{\sigma}{\rightarrow}(e',v')$ in the following.

\begin{definition}[Semantics of Timed Automaton]\label{def:sem_TA}
The semantics of a timed automaton given by the tuple $\TA=(E, E_{0}, C, \Sigma, I, \Delta)$ is the transition system $\TAsem=(S,S_{0},\Sigma\cup\mathds{R}_{\geq0}, T_{\tn{s}}\cup T_{\tn{d}})$, where $S$ is the set of states
\begin{align}
\notag S &= \{(e,v)\in E\times\mathds{R}^{C}_{\geq0}|~v\models I(e)\},
\end{align}
$S_0 \subseteq S$ is the set of initial states
\[S_{0} = \{(e,v)\in E_{0}\times \mathds{R}^{C}_{\geq0}|~v = v_0\}.\]
Note that $E\times\mathds{R}^{C}_{\geq0}$ induces subspace topology on $S$.

$T_{\tn{s}}\cup T_{\tn{d}}$ is the union of the following sets of transitions
\begin{align}
\notag T_{\tn{s}} &= \{(e,v)\overset{\sigma}{\rightarrow}(e',v')|~\exists (e,G_{e \rightarrow e'},\sigma,R_{e\rightarrow e'},e')\in\Delta 
\text{ such that }
v\models G_{e\rightarrow e'} \text{ and }v'=v[R_{e\rightarrow e'}]\},\\
\notag T_{\tn{d}} &= \{(e,v)\overset{d}{\rightarrow}(e,v+d)|~\forall d'\in[0,d] : v+d'\models I(e)\}.
\end{align}
\end{definition}
Hence, the semantics of a timed automaton is a transition system that comprises an infinite number of states: product of $E$ and $\mathds{R}^{C}_{\geq0}$ and two types of transitions: the transition set $T_{\tn{s}}$ between discrete states with possibly a reset of clocks belonging to a subset $R_{e\rightarrow e'}$, and the transition set $T_{\tn{d}}$ that corresponds to time passing within the invariant $I(e)$.

In the following, we define an analog to the solution of a dynamical system for a timed automaton.
\begin{definition}[Run of Timed Automaton]\label{def:run_of_timed_automaton}
A run of a timed automaton \TA is a possibly infinite sequence of alternations between time steps and discrete steps of the following form
\begin{align}
(e_{0},v_{0})\overset{d_{1}}{\longrightarrow}(e_{0},v_{1})\overset{\sigma_{1}}{\longrightarrow}(e_{1},v_{2})\overset{d_{2}}{\longrightarrow}\dots,
\label{eq:RunTimedAutomaton}
\end{align}
where $d_{i}\in \mathds{R}_{\geq0}$ and $\sigma_{i}\in \Sigma$.
\end{definition}
By forcing alternation of time and discrete steps in Definition~\ref{def:run_of_timed_automaton}, the time step $d_i$ is the maximal time step between the discrete steps $\sigma_{i-1}$ and $\sigma_{i}$. 

\subsection{Trajectory of a Timed Automaton}
A vital object for studying the behavior of any dynamical system is its trajectory. Therefore, we define a trajectory of a timed automaton \cite{Complete_Abstractions_of_Dynamical_Systems_by_Timed_Automata}. At the outset, we introduce a concept of a time domain.

In the following, we denote sets of the form
$\{a,\dots\}$ with $~a\in\mathds{Z}_{\geq0}$ as $\{a,\dots,\infty\}$. Let $k\in \mathds{N}\cup
\{\infty\}$; a
subset $\mathcal{T}_k \subset \mathds Z_{\geq0} \times \mathds R_{\geq0}$ with disjoint (union) topology will be
called a time domain if there exists an increasing sequence
$\{t_i\}_{i\in \{0,\hdots, k\}}$ in $\mathds R_{\geq0}\cup\{\infty\}$ such
that
\[
\mathcal{T}_k = \bigcup_{i \in \{1,\hdots, k\}}\{i\} \times T_i,
\]
where 
\[T_i = \left\{\begin{matrix}[t_{i-1}, t_{i}] & \hbox{ if } t_i <
    \infty \\ [t_{i-1},\infty[ & \hbox{ if } t_i =
    \infty.  \end{matrix} \right.\]

Note that $T_i= [t_{i-1}, t_{i}]$ for all $i$ if $k=\infty$.  We
say that the time domain is infinite if $k = \infty$ or $t_k =
\infty $. The sequence $\{t_i\}_{i\in \{0,\hdots, k\}}$ corresponding
to a time domain will be called a switching sequence.

We define two projections $\pi_1: E \times  \mathds{R}^{C}_{\geq0} \to E$ and $\pi_2: E \times  \mathds{R}^{C}_{\geq0} \to \mathds{R}^{C}_{\geq0}$ by  $\pi_1(e,v) = e$ and $\pi_2(e,v) = v$.

\begin{definition}[Trajectory]\label{def:trajectory}
A trajectory of the timed automaton \TA is a   pair $({\cal T}_k,\gamma)$ where $k\in \mathds{N}\cup \{\infty\}$ is fixed, and
\begin{itemize}
\item $\mathcal{T}_k \subset \mathds Z_{\geq0} \times \mathds R_{\geq0}$ is a time domain with corresponding switching sequence $\{t_i\}_{i\in \{0,\hdots, k\}}$,
\item $\gamma:\mathcal{T}_k \to S$ and recall that $S$ is the (topological) space of joint continuous and discrete states, see Definition~\ref{def:sem_TA} and Remark~\ref{rem:topology_valuation}. The map $\gamma$ satisfies:
\begin{enumerate}
\item For each $i \in \{1, \hdots, k-1\}$, there exists $\sigma \in \Sigma$ such that
\[\gamma(i,t_i) \overset{\sigma}{\longrightarrow} \gamma(i+1,t_i) \in T_{\tn{s}}.\]
\item Let $\bm{0}$ be a vector of zeros and $\bm{1}$ be a vector of ones in $\mathds{R}^{C}$. 
For each $i \in \{1, \hdots, k\}$
\begin{align}\notag\pi_2(\gamma(i,t_{i-1}+d))=\pi_2(\gamma(i,t_{i-1}))+d\bm{1}~~~~\forall d\in\begin{cases}[0,\,t_{i}-t_{i-1}]&\tn{ if } t_i <
    \infty\\
    [0,\,\infty[&\tn{ if } t_i =
    \infty\end{cases},\end{align}
    where $\pi_2(\gamma(i,t_{i-1}+d))\in \llbracket I(\pi_1(\gamma(i,t_i))\rrbracket$ and $\pi_2(\gamma(1,t_{0}))=\bm{0}$.
    (Item 2 ensures that the time derivative of the valuation of each clock is one, between the discrete transitions.)
\end{enumerate}
\end{itemize}
Note that $\gamma$ is continuous by construction. Recall the definition of $v_{0}$ from Definition~\ref{def:clock_valuation}. A trajectory at $(e,v_{0})$ (with $v_{0} \models I(e)$ ) is a trajectory $({\cal T}_k,\gamma)$ with $\gamma(1,t_0)=(e,v_{0})$.
\end{definition}

We define a discrete counterpart of the flow map.

\begin{definition}[Flow Map of Timed Automaton]
The flow map of a timed automaton \TA is a multivalued map
\[\phi_{\TA}:\mathds{R}_{\geq0}\times S_{0}\rightarrow 2^{S},\] defined by $ (e',v') \in\phi_{\TA}(t;e,v_{0})$  if and only if there exists a trajectory $(\mathcal T_k,\gamma)$ at $(e,v_{0})$ such that $t = t_{k}-t_0$ and $(e',v') = \gamma(k,t_k)$.
\end{definition}

It will be instrumental to define  a discrete flow map $\Phi_{\TA}:\mathds{R}_{\geq0}\times E_{0}\rightarrow 2^{E}$, which forgets the valuation of the clocks
\begin{align}
\Phi_{\TA}(t,e) = \pi_1 \circ \phi_{\TA}(t;e,v_0)
\end{align}

In other words, $\Phi_{\TA}$ is defined by: $e' \in\Phi_{\TA}(t,e)$ if and only if there exists a run \eqref{eq:RunTimedAutomaton} of  $\TAsem$ initialized in $(e,v_{0})$ that reaches the location $e'$ at time $t=\sum_{i} d_{i}$.


The reachable set of a timed automaton is defined as follows.
\begin{definition}[Reachable set of Timed Automaton]
The reachable locations of a system \TA from a set of initial locations $E_{0}\subseteq E$ on the time interval $[t_{1},t_{2}]$ is defined as
\begin{align}
\Phi_{\TA}([t_1, t_2], E_0) \equiv \bigcup_{(t,e) \in [t_1, t_2] \times E_0} \Phi_{\TA}(t, e).
\end{align}
\end{definition}


\end{document}